\documentclass{ifacconf}
\usepackage{amsmath}
\usepackage{amsfonts}
\usepackage{enumitem}
\usepackage{mathrsfs}
\usepackage[cal=boondoxo]{mathalpha}
\usepackage{graphicx}      
\usepackage{natbib}         

\newcommand\R{\ensuremath{\mathbb{R}}}

\newcommand\Mb{\ensuremath{\mathbf{M}}}

\newcommand\xb{\ensuremath{\mathbf{x}}}
\newcommand\ub{\ensuremath{\mathbf{u}}}
\newcommand\vb{\ensuremath{\mathbf{v}}}

\newcommand\yb{\ensuremath{\mathbf{y}}}

\newcommand\zb{\ensuremath{\mathbf{z}}}

\DeclareEmphSequence{\bfseries\itshape}
\let\epsilon\varepsilon

\begin{document}
\begin{frontmatter}

\title{
\centering
\parbox{\textwidth}{%
\centering
Data-Driven Probabilistic Fault Detection\\
and Identification via Density Flow Matching
}%
}


\thanks[footnoteinfo]{This research is funded in part by the Technology Innovation Institute and DARPA Learning Introspective Control (LINC).}

\author[First]{Joshua D. Ibrahim} 
\author[First]{Mahdi Taheri} 
\author[First]{Soon-Jo Chung} 
\author[First]{Fred Y. Hadaegh}

\address[First]{Division of Engineering and Applied Science, California Institute of Technology, 
   Pasadena, CA 91106 USA (e-mail: {jdibrahi, mtaheri, sjchung, hadaegh}@caltech.edu).}

\begin{abstract}                
Fault detection and identification (FDI) is critical for maintaining the safety and reliability of systems subject to actuator and sensor faults. In this paper, the problem of FDI for nonlinear control-affine systems under simultaneous actuator and sensor faults is studied. We model fault signatures through the evolution of the probability density flow along the trajectory and characterize detectability using the 2-Wasserstein metric. In order to introduce quantifiable guarantees for fault detectability based on system parameters and fault magnitudes, we derive upper bounds on the distributional separation between nominal and faulty dynamics. The latter is achieved through a stochastic contraction analysis of probability distributions in the 2-Wasserstein metric. A data-driven FDI method is developed by means of a conditional flow-matching scheme that learns neural vector fields governing density propagation under different fault profiles. To generalize the data-driven FDI method across continuous fault magnitudes, Gaussian bridge interpolation and Feature-wise Linear Modulation (FiLM) conditioning are incorporated. The effectiveness of our proposed method is illustrated on a spacecraft attitude control system, and its performance is compared with an augmented Extended Kalman Filter (EKF) baseline. The results confirm that trajectory-based distributional analysis provides improved discrimination between fault scenarios and enables reliable data-driven FDI with a lower false alarm rate compared with the augmented EKF.
\end{abstract}

\begin{keyword}
Fault detection and identification, data-driven, flow matching, contraction theory, density space.
\end{keyword}

\end{frontmatter}

\section{Introduction}
Fault detection and identification (FDI) is fundamental to ensuring the safety, reliability, and resilience of safety-critical systems. Applications such as spacecraft attitude control, autonomous ground and aerial vehicles, industrial automation, and robotic systems require rapid detection of actuator and sensor degradations to maintain closed-loop stability and prevent mission-critical failures \cite{10502204,ragan2024online}. Analytical redundancy methods, including parity relations, observer-based residual generation, and model-based FDI approaches, form the basis of most classical solutions \cite{chen2012robust, isermann2005fault}. These techniques have been widely adopted due to their interpretability and real-time feasibility, and remain essential in aerospace and automotive onboard FDI systems.

Despite the extensive literature on the problem of FDI, significant challenges remain when dealing with nonlinear systems under uncertainty. Classical and model-based FDI methods often assume linear dynamics or focus on isolated fault scenarios (i.e., non-simultaneous faults); hence, their performance degrades when facing simultaneous actuator and sensor faults, and stochastic disturbances \cite{venkatasubramanian2003review, patton2013issues}. Learning-based FDI methods address some of these challenges, but typically focus on reconstructing fault parameters along individual trajectories rather than characterizing how faults modify the underlying probability distributions of the system states \cite{bakhtiaridoust2023data, talebi2008recurrent}. 

On the other hand, tools from stochastic analysis, such as the Fokker-Planck equation (FPE), provide a principled description of the evolution of state densities under uncertainties \cite{pavliotis2014stochastic, risken1989fokker}. However, these tools have not been extensively incorporated into FDI frameworks. Moreover, contraction theory has emerged as a powerful methodology for incremental stability in nonlinear deterministic and stochastic systems \cite{Lohmiller1998,tsukamoto2020neural,tsukamoto2020stochneural,tsukamoto2021contraction, bouvrie2019wasserstein}, yet its implications for distinguishing nominal and faulty dynamics at the distributional level remain largely unexplored. Restricted to linear systems under Gaussian density assumptions,~\cite{eriksson2013method,jung2020sensor,fu2018evaluation} have developed density-based methods for detectability and distinguishability of faults. This work addresses these gaps by proposing a distribution-based framework for FDI in nonlinear stochastic systems subject to simultaneous actuator and sensor faults.  

\begin{figure}
\begin{center}
\includegraphics[width=9.0cm]{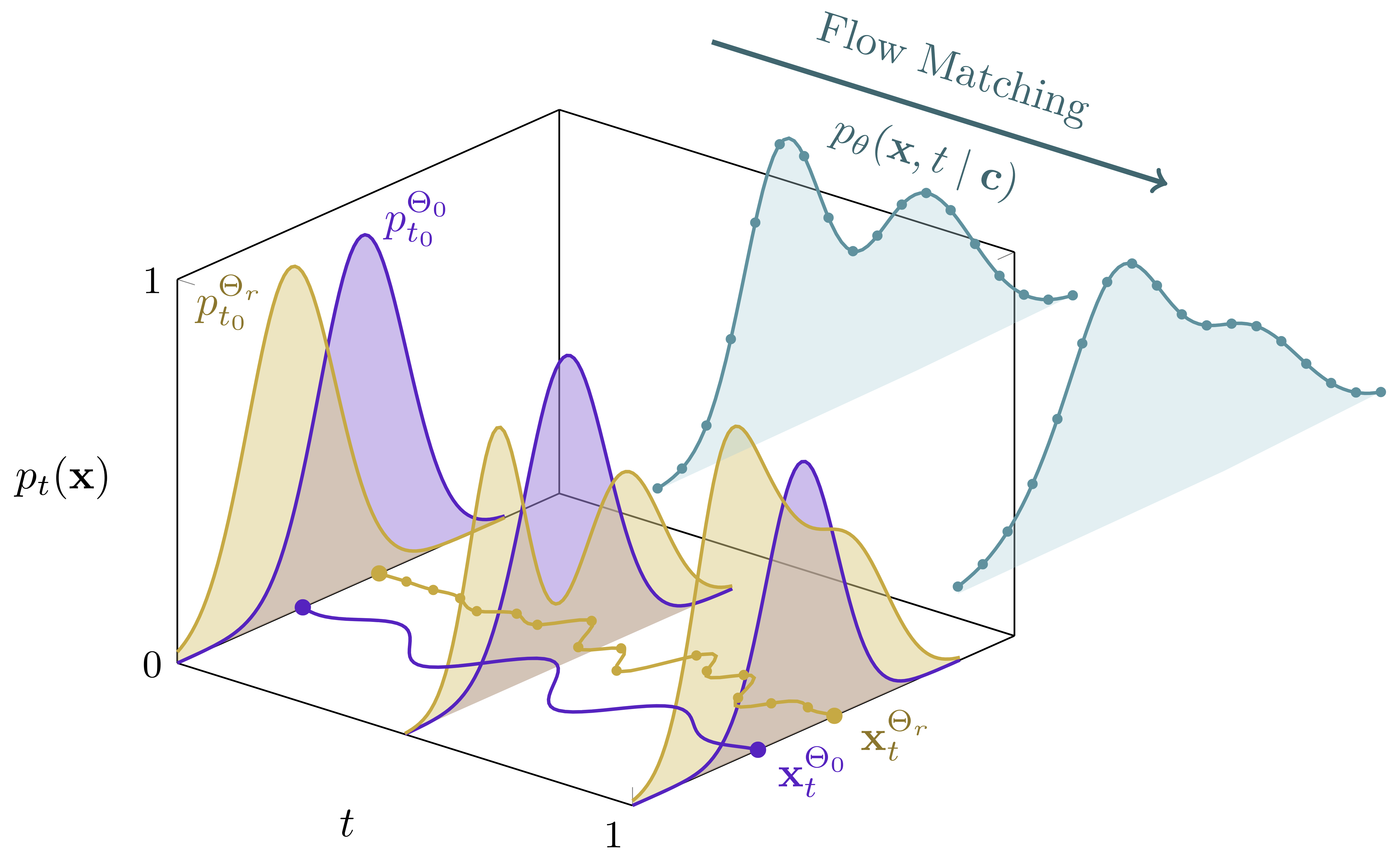}    
\caption{The illustration shows the relationship between the nominal trajectory $\xb^{\Theta_0}$ and the fault trajectory $\xb^{\Theta_r}$ and their induced densities of the nominal $p_t^{\Theta_0}$ and fault $p_t^{\Theta_r}$ systems (foreground) with the learned density $p_{\boldsymbol{\theta}}$ adapting toward the fault system (background).} 
\label{fig:traj_flow_matching}
\end{center}
\end{figure}

In this work, we propose a probabilistic framework to address the problem of FDI that operates on trajectory-level probability densities rather than individual state measurements. The core idea is to characterize each fault scenario by how it shapes the distribution of system trajectories over time. Under nominal conditions, an ensemble of initial states evolves according to a baseline probability density governed by the FPE, while under faulty conditions, the same ensemble follows a perturbed density flow that diverges from the nominal behavior. By comparing the densities of the dynamics, we can characterize fault profiles in a manner more robust than single-trajectory analysis. Using the Wasserstein distance allows for additional consideration of the geometric discrepancy between the density flows (\cite{villani2008optimal}), together leading to improved discrimination and greater reliability in fault detection and identification.

In order to utilize the availability of data in modern systems, we develop a data-driven approach based on a neural network-based conditional flow matching methodology that learns to predict how state densities evolve under different fault conditions. Specifically, we train a neural network conditioned on fault parameters to model the transition dynamics between consecutive states in trajectories. When the trained neural network is deployed, given an observed trajectory with unknown faults, we estimate the fault parameters by minimizing the trajectory-level negative log-likelihood across all candidate fault profiles, effectively inverting the learned density flows to infer which fault scenario best explains the observations. 

The aforementioned neural network-based inference framework enables generalization to unseen fault combinations via continuous interpolation in the fault parameter space, which allows the model to estimate fault magnitudes that were not explicitly present in the training data. The gradient-based optimization during inference (i.e., neural network deployment) exploits the smoothness of the learned density model, converging to accurate fault parameter estimates within a few hundred iterations without requiring explicit retraining or combinatorial search over discrete fault hypotheses.

To summarize, the main contributions of this work are as follows. First, we establish formal theoretical guarantees on fault detectability and identifiability by extending contraction theory to the distributional setting, quantifying the separation between nominal and faulty system behaviors in the 2-Wasserstein metric. We formally provide an explicit upper bound on the Wasserstein distance between probability measures under different fault profiles, which enables rigorous assessment of when and how faults can be detected and distinguished through their differences in trajectory-level density flows. 

Second, we develop a data-driven FDI framework based on conditional flow matching that learns neural vector fields to characterize the evolution of probability densities under different fault scenarios. By leveraging Gaussian bridge interpolation and Feature-wise Linear Modulation (FiLM) conditioning, our approach enables generalization across fault profiles while maintaining computational tractability through trajectory-level negative log-likelihood optimization. 
\section{Problem Formulation}
\subsection{System Model} 
Let us consider the following nonlinear system under actuator and sensor faults
\begin{equation}
\label{eq:full_system_deterministic}
\begin{split}
\dot{\xb} &= f(\xb, t) + G(\xb, t)\Xi\ub, \\
\yb &= \Gamma h(\xb,t),
\end{split}
\end{equation}
where $\mathbf{x}: \mathbb{R}_{\geq 0} \rightarrow \mathbb{R}^n$ is the state, $\mathbf{u}: \mathbb{R}^n \times \mathbb{R}_{\geq 0} \rightarrow \mathbb{R}^m$ is the control input, $\mathbf{y}: \mathbb{R}_{\geq 0} \rightarrow \mathbb{R}^p$ is the output. Moreover, $f: \mathbb{R}^n \times \mathbb{R}_{\geq 0} \rightarrow \mathbb{R}^n, t \in \mathbb{R}_{\geq 0}$, $G: \mathbb{R}^n \times \mathbb{R}_{\geq 0} \rightarrow \mathbb{R}^{n \times m}$, and $h: \mathbb{R}^n \times \mathbb{R}_{\geq 0} \rightarrow \mathbb{R}^p$ are known smooth functions. The unknown matrices $\Xi = \operatorname{diag}(\eta_1(t), \ldots, \eta_m(t)) \in \mathbb{R}^{m \times m}$ and $\Gamma = \operatorname{diag}(\gamma_1(t), \ldots, \gamma_p(t)) \in \mathbb{R}^{p \times p}$ denote multiplicative actuator and sensor faults, respectively, with $\eta_i, \gamma_j \in [0,1]$, for $i=1,\ldots,m$ and $j=1,\ldots,p$. The parameters $\eta_i$ and $\gamma_j$ represent the level of degradation (or effectiveness) of the $i$-th actuator and the $j$-th sensor, respectively, where $\eta_i = 1$ and $\gamma_j = 1$ indicate nominal operation, while $\eta_i = 0$ and $\gamma_j = 0$ indicate complete failure. We denote the $r$-th possible instance of actuator and sensor faults for $r \in \{0,\ldots, N_f\}$ as fault profile $\Theta_r = (\boldsymbol{\eta}, \boldsymbol{\gamma}) \in \mathcal{F}$ where $\mathcal{F}:= [0,1]^m \times [0,1]^p$. We denote  $\Theta_0$ as the nominal, i.e., fault-free, behavior where $\Theta_0 = (\mathbf{1}_m, \mathbf{1}_p)$.  

When necessary, we denote $\mathbf{F}_{\Theta_r}(\xb,t):=f(\xb, t) + G(\xb, t)\Xi\ub_{cl}$ where $\ub = \ub_{cl}$ is some closed-loop controller used to drive the system to some desired behavior. The nonlinear system~\eqref{eq:full_system_deterministic} in the presence of stochastic disturbances can be described by the Itô stochastic differential equation (SDE) in the following form:
\begin{equation}
\label{eq:full_system_stochastic}
\begin{split}
d \mathbf{x} &= \mathbf{F}_{\Theta_r}(\xb,t) d t + \sigma(\mathbf{x}, t) d \mathscr{W}_1(t), \\
\mathbf{y}dt &= \Gamma h(\mathbf{x}, t)dt + d \mathscr{W}_2(t),
\end{split}
\end{equation}
where $\sigma: \mathbb{R}^n \times \mathbb{R}_{\geq 0} \rightarrow \mathbb{R}^{n \times d}$ and $\mathscr{W}_1(t), \ \mathscr{W}_2(t)$ are standard Wiener processes of dimensions $d$ and $p$, respectively. The stochastic system \eqref{eq:full_system_stochastic} can be seen as the \emph{perturbed} version of~\eqref{eq:full_system_deterministic}. We assume that $\exists L_1, L_2 > 0$ such that
\begin{equation}
\label{eq:lipschitz_growth_conditions}
\begin{split}
 \|\mathbf{F}_{\Theta_r}(\xb,t)-\mathbf{F}_{\Theta_r}(\yb,t)\| &+ \|\sigma(\mathbf{x}, t)-\sigma(\mathbf{y}, t)\| \leq L_1\|\mathbf{x}-\mathbf{y}\|, \\
 \|\mathbf{F}_{\Theta_r}(\xb,t)\|^2 &+ \|\sigma(\mathbf{x}, t)\|_F^2 \leq L_2(1+\|\mathbf{x}\|^2),
\end{split}
\end{equation}
$\forall \xb, \yb, t$, which ensure the existence and uniqueness of strong solutions to~\eqref{eq:full_system_stochastic} (see~\cite{karatzas2014brownian}). Moreover, $\|\cdot\|$ and $\|\cdot\|_F$ denote the Euclidean and Frobenius norms, respectively.

\subsection{Density Propagation}
Along with the deterministic and stochastic formulations of the system dynamics, one can also describe the concentration of states for an ensemble of trajectories from the evolution of the state probability density function (PDF). Formally, let $\mu_t, \nu_t \in \mathscr{P}_2(\R^n)$ be the time-evolving probability measures (laws) of the state and $\mathscr{P}_2(\mathbb{R}^n)$ be the space of probability measures on $\mathbb{R}^n$ with finite second moments. We assume that $\mu_t$ is absolutely continuous w.r.t the Lebesgue measure which implies by the Radon–Nikodym theorem that there exists a density $p(\xb,t)\in L^1(\R^n, d\xb)$ such that $d\mu_t(\xb) = p(\xb, t)d\xb$ (see~\cite{pavliotis2014stochastic}) where $L^1$ is the Banach space of integrable functions. 

For the unperturbed system~\eqref{eq:full_system_deterministic}, the density $p(\xb, t) \geq 0$ with initial condition $p(\xb,0) = p_0(\xb)$ evolves according to the Liouville equation as
\begin{equation}
\label{eq:liouville_equation}
\frac{\partial p}{\partial t} = -\nabla_{\xb} \cdot \left(p(\xb,t) \mathbf{F}_{\Theta_r}(\xb,t)\right),
\end{equation}
where $p:\R^n \times \R_{\geq 0} \to \R_{\geq 0}$ and $\nabla_{\xb} \cdot:= \sum_{i=1}^n \frac{\partial}{\partial x_i}$ is the divergence operator w.r.t $\xb$. For the perturbed system~\eqref{eq:full_system_stochastic}, the density $p(\xb, t)$ evolves according to the Fokker-Planck equation (FPE) as
\begin{equation}
\label{eq:fokker_planck_equation}
\frac{\partial p}{\partial t} = -\nabla_{\xb} \cdot (p(\xb,t) \mathbf{F}_{\Theta_r}(\xb,t)) + \frac{1}{2} \nabla_{\xb}\cdot \nabla_{\xb}\cdot (p(\xb,t) D(\xb,t)),
\end{equation} 
where $D(\xb,t) = \sigma(\xb,t)\sigma(\xb,t)^\top$ is known as the diffusion matrix and $\nabla_{\xb}\cdot \nabla_{\xb}\cdot$ is the second order differential operator defined as $\nabla_{\xb}\cdot \nabla_{\xb}\cdot := \sum_{i=1}^n \sum_{j=1}^n \frac{\partial^2}{\partial x_i \partial x_j}$. The Liouville equation~\eqref{eq:liouville_equation} is a special case of the FPE~\eqref{eq:fokker_planck_equation} with zero diffusion $D(\xb,t) = 0$. 

\subsection{Problem Statement}
This paper addresses the problem of fault detection and identification (FDI) for nonlinear stochastic systems subject to simultaneous actuator and sensor faults with stochastic disturbances. We consider the system given by~\eqref{eq:full_system_stochastic}, where $\eta_i \in[0,1]$ and $\gamma_j \in[0,1]$ are multiplicative actuator and sensor faults in the fault profile $\Theta = (\boldsymbol{\eta},\boldsymbol{\gamma})$, respectively. Rather than analyzing individual trajectories, we approach FDI through the evolution of probability density functions. In particular, we consider an ensemble of initial conditions $\xb(0) \sim p_0(\xb)$ under fault profile $\Theta_r$ that generates a time-varying density $p(\xb, t)$ and evolves according to \eqref{eq:fokker_planck_equation}. We aim to develop a data-driven framework that detects faults by quantifying distributional separation between nominal and faulty behaviors via the $2$-Wasserstein metric $\mathbb{W}_2\left(\mu_{t_0: t_N}^{\Theta_r}, \mu_{t_0: t_N}^{\Theta_0}\right)$, identifies fault profiles by learning probability flows through conditional flow matching with a network $\mathcal{F}_\theta$ conditioned on $\mathbf{c}=\left[\boldsymbol{\eta}^{\top}, \boldsymbol{\gamma}^{\top}\right]^{\top}$, and estimates the fault magnitudes $\eta_i$ and $\gamma_j$ by minimizing the negative log-likelihood over the observed trajectory.

\section{Fault Detectability and Identifiability in the Density Space}
Considering the fault profile $\Theta_r$ for the perturbed system \eqref{eq:full_system_stochastic}, the density $p(\xb,t)$ is the solution to~\eqref{eq:fokker_planck_equation} with the initial condition $p(\xb,0) = p_0(\xb)$ at time $t$. For any finite horizon $\{t_0,\dots, t_N\}$, this induces the joint density $p_{t_0:t_N}(\xb_{t_0},\ldots,\xb_{t_N})\in L^1\left((\mathbb{R}^n)^{(N+1)}, d \xb_{t_0} \cdots d \xb_{t_N}\right)$.  This is the space of joint densities over $N+1$ timepoints. We denote $p_{t_0:t_N}^{\Theta_r}$ as the the joint density of states at times $\{t_0,\dots, t_N\}$ under fault profile $\Theta_r$.

Evaluating the joint density over the full horizon rather than at a single time instance allows one to compare entire trajectory signatures and consequently expose sharper distinctions between different fault profiles. Intuitively, a single-time snapshot $p_t$ may not distinguish fault profiles if their state distributions overlap significantly at that instant. However, the entire trajectory signature $p_{t_0:t_N}$ encodes how the system evolves dynamically, amplifying differences between fault modes through cumulative 
divergence over time. To quantify the degree of detectability and identifiability between various fault profiles, we define measures based on distance metrics to compute the separation between densities. 
\begin{defn}[$\epsilon$-Detectability]
Fix any statistical distance metric $\mathscr{D}:\mathscr{P}_2(\mathbb{R}^{n(N+1)}) \times \mathscr{P}_2(\mathbb{R}^{n(N+1)}) \to \mathbb{R}_{\geq 0}$, where $\mathscr{P}_2(\cdot)$ denotes the space of probability measures. For any fault profile $\Theta_r \in \mathcal{F}$ such that $ \Theta_r \neq \Theta_0$ and $r \in \{1,\ldots, N_f\}$, the fault profiles are $\epsilon$-detectable on the horizon $\{t_0, \ldots, t_N\}$ if 
\begin{equation}   
\mathscr{D}\left(\mu_{t_0:t_N}^{\Theta_r}, \mu_{t_0:t_N}^{\Theta_0}\right) \geq \epsilon ,
\end{equation}
where $\mu_{t_0: t_N}^{\Theta_r}$ and $\mu_{t_0: t_N}^{\Theta_0}$ are the joint probability measures associated with joint densities $p_{t_0:t_N}^{\Theta}$ and $p_{t_0:t_N}^{\Theta_0}$ over the horizon $\{t_0, \ldots, t_N\}$ under fault profile $\Theta_r$ and nominal profile $\Theta_0$, respectively, and $\epsilon > 0$.
\end{defn}
\begin{defn}[$\epsilon$-Identifiability]
Fix any statistical distance metric $\mathscr{D}:\mathscr{P}_2(\mathbb{R}^{n(N+1)}) \times \mathscr{P}_2(\mathbb{R}^{n(N+1)}) \to \mathbb{R}_{\geq 0}$. For any two fault profiles $\Theta_\mathcal{i}, \Theta_\mathcal{j} \in \mathcal{F}$, where $\Theta_\mathcal{i} \neq \Theta_\mathcal{j}$ and $\mathcal{i},\mathcal{j}\in\{1,\ldots, N_f\}$, the fault profiles are $\epsilon$-identifiable on the horizon $\{t_0, \ldots, t_N\}$ if
\begin{equation}   
\mathscr{D}\left(\mu_{t_0:t_N}^{\Theta_\mathcal{i}}, \mu_{t_0:t_N}^{\Theta_\mathcal{j}}\right) \geq \epsilon ,
\end{equation}
where $\mu_{t_0: t_N}^{\Theta_\mathcal{i}}$ and $\mu_{t_0: t_N}^{\Theta_\mathcal{j}}$ are the joint probability measures associated with joint densities $p_{t_0:t_N}^{\Theta_\mathcal{i}}$ and $p_{t_0:t_N}^{\Theta_\mathcal{j}}$ over the horizon $\{t_0, \ldots, t_N\}$ under fault profile $\Theta_\mathcal{i}$ and $\Theta_\mathcal{j}$, respectively, and $\epsilon > 0$.
\end{defn}

Since there are various distance metrics defined on the space of probability measures, the choice of metric $\mathscr{D}$ can be made based on computational tractability and application requirements. In this work, we consider the 2-Wasserstein distance (see Definition~\ref{def:2w}) as our statistical distance metric $\mathscr{D}$ and compare it to other metrics like the Kullback-Leibler divergence. For the specific choice of 2-Wasserstein distance, one can quantify $\epsilon$ by performing the contraction analysis between \eqref{eq:full_system_deterministic} under $\Theta_0$ and~\eqref{eq:full_system_stochastic} when a fault profile $\Theta_r$ is active. The latter is investigated in the next section.

\section{Contraction Theory in Density Space}
Contraction theory is a tool for analyzing the incremental stability of nonlinear systems by studying the convergence of neighboring trajectories. In this section, we first review contraction analysis for deterministic systems and stochastic systems. Consequently, we extend the contraction analysis to distributions in the 2-Wasserstein metric. 
\subsection{Contraction of Trajectories}
Consider the closed-loop dynamics of~\eqref{eq:full_system_deterministic} under nominal profile $\Theta_0$, i.e.,
\begin{equation}
\label{eq:diff_ode}    
\dot{\xb} = \mathbf{F}_{\Theta_0}(\xb,t), \ \delta \dot{\xb} = \frac{\partial \mathbf{F}_{\Theta_0}}{\partial \mathbf{x}} \delta \mathbf{x}, \ \xb(0) = \xb_0,
\end{equation}
where $\delta\xb:=\xb-\xb_d$ is the virtual displacement of any two trajectories, $\xb_d$ is some desired trajectory, and $\mathbf{F}_{\Theta_0}(\xb,t)$ is the closed-loop vector field under some nominal controller $\ub=\ub_{cl}$. The dynamics~\eqref{eq:diff_ode} is contracting if all trajectories converge exponentially to each other (see \cite{Lohmiller1998} for more details). 
\begin{lem}[Deterministic Contraction]
\label{lem:deterministic_contraction}
System~\eqref{eq:diff_ode} is contracting (i.e., all the solution trajectories exponentially converge to a single trajectory globally from any initial condition), if there exists a uniformly positive definite metric $\Mb(\xb, t)=\mathbf{H}(\xb, t)^{\top} \mathbf{H}(\xb, t), \ \Mb(\xb, t) \succ 0, \forall \xb, t$, with a smooth coordinate transformation of the virtual displacement $\delta \zb=\mathbf{H}(\xb, t) \delta \xb$, such that
\begin{equation}
\dot{\Mb}(\xb, t)+2 \operatorname{sym}\left(\Mb(\xb, t) \frac{\partial \mathbf{F}_{\Theta_0}}{\partial \xb}\right) \preceq-2 \alpha \Mb(\xb, t), \quad \forall \xb, t,
\end{equation}
where $\alpha>0$. If \eqref{eq:diff_ode} is contracting, one has
\begin{equation}
\|\delta \zb(t)\|=\|\mathbf{H}(\xb, t) \delta \xb(t)\| \leq\|\delta \zb(0)\| e^{-\alpha t} .
\end{equation}
\begin{pf}
    See~\cite{Lohmiller1998}.\qed
\end{pf}
\end{lem}

Similarly, the closed-loop perturbed system \eqref{eq:diff_ode} under the fault profile $\Theta_0$ and a stochastic disturbance is described by
\begin{equation}
\label{eq:sde}
d \mathbf{x}=\mathbf{F}_{\Theta_0}\left(\mathbf{x}, t\right) d t+\sigma\left(\mathbf{x}, t\right) d \mathscr{W}(t), \ \xb(0) = \xb_0,
\end{equation}
where we utilize the same nominal controller $\ub=\ub_{cl}$. Considering \eqref{eq:lipschitz_growth_conditions}, there exists a unique strong solution to~\eqref{eq:sde}. Suppose $\mathbf{a}(t), \mathbf{b}(t)$ are two solutions to~\eqref{eq:sde} driven by independent Wiener processes $\mathscr{W}^{(1)}(t), \ \mathscr{W}^{(2)}(t)$, starting from initial conditions independent of the noise with initial distributions $\mathbf{a}(0)\sim\mu_0$ and $\mathbf{b}(0) \sim \nu_0$, respectively. One has
\begin{equation*}
d \mathbf{z}=\left[\begin{array}{l}\mathbf{F}_{\Theta_0}\left(\mathbf{a}, t\right) \\ \mathbf{F}_{\Theta_0}\left(\mathbf{b}, t\right)\end{array}\right] d t+\left[\begin{array}{cc}\sigma_1\left(\mathbf{a}, t\right) & 0 \\ 0 & \sigma_2\left(\mathbf{b}, t\right)\end{array}\right]\left[\begin{array}{l}d \mathscr{W}^{(1)} \\ d \mathscr{W}^{(2)}\end{array}\right],
\end{equation*}
where $\zb(t) = \left[\mathbf{a}(t)^\top, \mathbf{b}(t)^\top\right]^\top \in \R^{2n}$, and $\mu_t$ and $\nu_t$ are the marginals of $\mathbf{a}(t)$ and $\mathbf{b}(t)$, respectively. Consider the state-dependent Riemannian metric $\Mb(\xb(\lambda,t), t)$ for $\lambda \in [0,1]$, where $\xb(\lambda,t)$ is a geodesic curve connecting $\mathbf{a}(t)$ and $\mathbf{b}(t)$ such that $\xb(0,t) = \mathbf{a}(t)$ and $\xb(1,t) = \mathbf{b}(t)$. Moreover, $\sigma_1(\mathbf{a}, t)$ and $\sigma_2(\mathbf{b}, t)$ are defined as $\sigma(\xb(0, t), t) = \sigma_1(\mathbf{a}, t)$ and $\sigma(\xb(1,t), t) = \sigma_2(\mathbf{b}, t)$.

\subsection{Contraction of Distributions} 
Contraction theory provides formal guarantees about the behavior of states of a nonlinear system. However, to analyze the distribution of states, a metric that relates distances in a distributional sense must be employed. To achieve this, we consider the $2$-Wasserstein distance between $\mu_t$ and $\nu_t$ as it captures the underlying geometry of the probability flow, remains finite even for non-overlapping supports, and captures differences in location, spread, and shape that KL-type divergences fail to reflect (see~\cite{arjovsky2017wasserstein,villani2008optimal}). 

\begin{defn}[2-Wasserstein Distance]\label{def:2w}
Let $\mathscr{P}_2(\mathbb{R}^n)$ be the space of probability measures on $\mathbb{R}^n$ with finite second moments. Let $\Pi(\mu, \nu)$ be the set of the joint probability measures $\pi$ with marginals $\mu$ and $\nu$. The $2$-Wasserstein metric between two measures $\mu, \nu \in \mathscr{P}_2(\mathbb{R}^n)$ is defined as
\begin{equation*}
\begin{split}
\mathbb{W}_2(\mu, \nu) &= \left( \inf_{\pi \in \Pi(\mu, \nu)} \int_{\mathbb{R}^d \times \mathbb{R}^d} \|\xb-\yb\|^2 \, d\pi(\xb,\yb) \right)^{1/2} \\
&= \left(\inf_{\pi \in \Pi(\mu, \nu)} \mathbb{E}_{\pi}[\|\xb-\yb\|^2]\right)^{1/2},
\end{split}
\end{equation*} 
such that $(\mathbf{x},\mathbf{y}) \sim \pi \in \Pi(\mu,\nu)$ and the infimum is taken over the set of all joint measures on $\R^n \times \R^n$ with marginals $\mu$ and $\nu$. 
\end{defn}

\begin{thm}[Wasserstein Contraction]
\label{thm:wasserstein_contraction}
Suppose that there exist bounded positive constants $\underline{m}, \overline{m}, g_1, g_2, \overline{m}_\xb$, and $\overline{m}_{\xb^2}$, such that $\underline{m} \leq\|\Mb(\xb, t)\| \leq \overline{m}$, $\left\|\sigma_1(\xb, t)\right\|_F \leq~g_1$, $ \left\|\sigma_2(\xb, t)\right\|_F \leq g_2,\ \left\|\partial\left(\Mb_{i j}\right) / \partial \xb\right\| \leq \overline{m}_\xb$, and 
$$\left\|\partial^2\left(\Mb_{i j}\right) / \partial \xb^2\right\| \leq \overline{m}_{\xb^2}, \forall \xb, t.$$ 
Suppose that the hypothesis of Lemma~\ref{lem:deterministic_contraction} holds, i.e., the deterministic system~\eqref{eq:diff_ode} is contracting. Consider the generalized squared length with respect to a Riemannian metric $\Mb(\xb(\lambda, t), t)$ defined by
\begin{equation}
V(\xb, \delta \xb, t)=\delta \xb^\top \Mb(\xb, t) \delta \xb=\int_0^1 \frac{\partial \xb}{\partial \lambda}^{\top} \Mb(\xb, t) \frac{\partial \xb}{\partial \lambda} d \lambda ,
\end{equation}
where $\delta \mathbf{x} = \mathbf{a} - \mathbf{b}$ is the virtual displacement, such that $\underline{m}\|\mathbf{a}-\mathbf{b}\|^2 \leq V(\mathbf{x}, \delta \xb, t) \leq \overline{m}\|\mathbf{a}-\mathbf{b}\|^2$ holds. One has
\begin{equation}
\label{eq:mean_squared}
    \mathbb{E}\left[\left\|\mathbf{a}(t)-\mathbf{b}(t)\right\|^2\right] \leq \frac{C_c}{2 \gamma_1}+\frac{1}{\underline{m}}\mathbb{E}\left[V\left(\xb(0), \delta \xb(0), 0\right)\right] e^{-2 \gamma_1 t} ,
\end{equation}
for $\gamma_1=\alpha-\left(\left(g_1^2+g_2^2\right) / 2 \underline{m}\right)\left(\varepsilon_c \overline{m}_\xb+\overline{m}_{\xb^2} / 2\right)$ and $C_c= \left(\overline{m} / \underline{m}+\overline{m}_\xb /\left(\varepsilon_c \underline{m}\right)\right)\left(g_1^2+g_2^2\right)$, where $\alpha$ is the contraction rate of~\eqref{eq:diff_ode} and $\epsilon_c >0$ is a constant. In the 2-Wasserstein metric, one has
\begin{equation}
    \mathbb{W}_2^2(\mu_t, \nu_t) \leq \frac{C_c}{2 \gamma_1}+\frac{\overline{m}}{\underline{m}}\mathbb{W}_2^2(\mu_0, \nu_0) e^{-2 \gamma_1 t} .
\end{equation}
\end{thm}
\begin{pf}
Following the computation of the infinitesimal generator $\mathscr{L}$ of $V$, we have $\mathscr{L}V \leq -2\gamma_1 V + \underline{m}C_c$. Taking the expectations of both sides with Dynkin's formula yields the desired mean-squared bound. From the upper bounds $V(\mathbf{x}, \delta \xb, t) \leq \overline{m}\|\mathbf{a}-\mathbf{b}\|^2$ and $\mathbb{E}[V(\xb(0), \delta \xb(0), 0)] \leq \overline{m}\mathbb{E}[\|\mathbf{a}(0) - \mathbf{b}(0)\|^2]$ the mean-squared bound in~\eqref{eq:mean_squared} can be derived as
\begin{equation}
    \mathbb{E}\left[\left\|\mathbf{a}(t)-\mathbf{b}(t)\right\|^2\right] \leq \frac{C_c}{2 \gamma_1}+\frac{\overline{m}}{\underline{m}}\mathbb{E}\left[\|\mathbf{a}(0) - \mathbf{b}(0)\|^2\right] e^{-2 \gamma_1 t}.
\end{equation}
See~\cite{tsukamoto2021contraction} for further details. Since $\mathbb{W}_2^2(\mu_t, \nu_t) \leq \mathbb{E}\left[\left\|\mathbf{a}(t)-\mathbf{b}(t)\right\|^2\right]$ holds for any $(\mathbf{a}(0),\mathbf{b}(0))$ with marginals $\mu_0$ and $\nu_0$, taking the infimum over all such couplings yields the desired bound for the 2-Wasserstein (squared) distance as
\begin{equation*}
\mathbb{W}_2^2(\mu_t, \nu_t) \leq \frac{C_c}{2 \gamma_1}+\frac{\overline{m}}{\underline{m}}e^{-2 \gamma_1 t}\mathbb{W}_2^2(\mu_0, \nu_0) . \qed
\end{equation*} 
\end{pf}

By choosing $\nu_0$ equal to the stationary distribution of~\eqref{eq:sde}, we see that Theorem~\ref{thm:wasserstein_contraction} suggests exponentially fast convergence of $\mu_t$ to a ball with the radius $\sqrt{\frac{C_c}{2 \gamma_1}}$.

\begin{thm}[Wasserstein FDI]
\label{thm:wasserstein_fdi}
Let the nominal system be described by~\eqref{eq:full_system_deterministic} as $\dot{\xb} = \mathbf{F}_{\Theta_0}(\xb,t)$ and $\Theta_0 = (\mathbf{1}_m, \mathbf{1}_p)$ which we will denote by trajectory $\mathbf{a}(t)$. Let the perturbed system be described by~\eqref{eq:full_system_stochastic}  as $d \mathbf{x} = \mathbf{F}_{\Theta_r}(\xb,t) d t + \sigma(\mathbf{x}, t) d \mathscr{W}_1(t)$ and $\Theta_r \neq (\mathbf{1}_m, \mathbf{1}_p)$ which we will denote with trajectory $\mathbf{b}(t)$, i.e.,
\begin{equation*}
d \mathbf{z}=\left[\begin{array}{c}
f(\mathbf{a}, t) + G(\mathbf{a},t)\ub \\
f(\mathbf{b}, t) + G(\mathbf{b},t)\Xi\ub
\end{array}\right] d t+\left[\begin{array}{cc}
0 \\
\sigma_2(\mathbf{b}, t)
\end{array}\right]\left[\begin{array}{l}
0 \\
d \mathscr{W}
\end{array}\right] ,
\end{equation*}
where $\mathbf{z}(t)=\left[\mathbf{a}(t)^{\top}, \mathbf{b}(t)^{\top}\right]^{\top} \in \mathbb{R}^{2 n}$. Moreover, $\mu_t$ and $\nu_t$ are the measures of $\mathbf{a}(t)$ and $\mathbf{b}(t)$, respectively. The state-dependent Riemannian metric $\mathbf{M}(\mathbf{x}(\lambda, t), t)$ for $\lambda \in[0,1]$ and $\mathbf{x}(\lambda, t)$ is a geodesic curve connecting $\mathbf{a}(t)$ and $\mathbf{b}(t)$ such that $\mathbf{x}(0, t)=\mathbf{a}(t)$ and $\mathbf{x}(1, t)=\mathbf{b}(t)$ with
\begin{equation*}
    V(\mathbf{x}, \delta \mathbf{x}, t)=\delta \mathbf{x}^{\top} M(\mathbf{x}, t) \delta \mathbf{x}=\int_0^1 \frac{\partial \mathbf{x}^{\top}}{\partial \lambda} M(\mathbf{x}(\lambda, t), t) \frac{\partial \mathbf{x}}{\partial \lambda} d \lambda.
\end{equation*} 
Let $\underline{m} \leq\|\Mb(\xb, t)\| \leq \overline{m}$,  $\|\sigma_2(\xb, t)\|_F \leq g_2$, $\left\|\partial \Mb_{i j} / \partial \xb\right\| \leq \overline{m}_\xb$, $\left\|\partial^2 \Mb_{i j} / \partial \xb^2\right\| \leq \overline{m}_{\xb^2}$ as per Theorem~\ref{thm:wasserstein_contraction}, but in addition uniform bounds on $\|G(\xb, t)\| \leq \bar{G}$, $\|\mathbf{u}(t)\| \leq \bar{\ub}$, and $\|\Xi-I\| \leq \bar{\Delta}$. One has
\begin{equation}\label{eq:W2_FDI}
\mathbb{W}_2^2\left(\mu_t, \nu_t\right) \leq \frac{\tilde C_c}{2 \tilde\gamma_1}+\frac{\tilde C_d}{2\tilde\gamma_1 \underline{m}}+\frac{\overline{m}}{\underline{m}}  \mathbb{W}_2^2\left(\mu_0, \nu_0\right) e^{-2 \tilde\gamma_1 t},
\end{equation}
where $\tilde \gamma_1=\tilde \alpha-\epsilon_f/4 - \left( g_2^2 / 2 \underline{m}\right)\left(\varepsilon_c \overline{m}_{\mathbf{x}}+\overline{m}_{\mathbf{x}^2} / 2\right)$ is a positive constant, $\tilde C_d=2\overline{m}^2\bar{G}^2\bar{\Delta}^2\bar{\mathbf{u}}^2/(\varepsilon_f\underline{m})>0$ is a constant (a function of $\bar{G}$, $\bar{\Delta}$, and $\bar{\ub}$), $\tilde\alpha$ is the contraction rate of $\dot{\xb} = f(\xb,t) + G(\xb,t)\ub$, and $$\tilde C_c= \left(\overline{m} / \underline{m}+\overline{m}_{\mathbf{x}} /\left(\varepsilon_c \underline{m}\right)\right)\left(g_2^2\right).$$
\end{thm}
\begin{pf}
To compute the infinitesimal generator $\mathscr{L}V$ for the coupled system, we decompose the drift term as:
\begin{multline*}
\frac{\partial V}{\partial \mathbf{a}} \cdot \mathbf{F}_{\Theta_0}(\mathbf{a}, t) + \frac{\partial V}{\partial \mathbf{b}} \cdot \mathbf{F}_{\Theta_r}(\mathbf{b}, t) = \\
\frac{\partial V}{\partial \mathbf{a}} \cdot \mathbf{F}_{\Theta_0}(\mathbf{a}, t) + \frac{\partial V}{\partial \mathbf{b}} \cdot \mathbf{F}_{\Theta_0}(\mathbf{b}, t) + \frac{\partial V}{\partial \mathbf{b}} \cdot [\mathbf{F}_{\Theta_r}(\mathbf{b}, t) - \mathbf{F}_{\Theta_0}(\mathbf{b}, t)].
\end{multline*}
The nominal dynamics contribution yields $-2\tilde{\alpha}V$ by the contraction of $\dot{\xb} = f(\xb, t)+G(\xb,t)\ub$. For the fault term, note that $\mathbf{F}_{\Theta_r}(\mathbf{b}, t) - \mathbf{F}_{\Theta_0}(\mathbf{b}, t) = G(\mathbf{b}, t)(\Xi - I)\mathbf{u}$. Using $\|\partial V/\partial \mathbf{b}\| \leq 2\overline{m}\|\delta\mathbf{x}\|$ and $\|\delta\mathbf{x}\| \leq \sqrt{V/\underline{m}}$, we have:
\begin{multline*}
\left|\frac{\partial V}{\partial \mathbf{b}} \cdot [G(\mathbf{b}, t)(\Xi-I)\mathbf{u}]\right| \leq 2\overline{m}\sqrt{\frac{V}{\underline{m}}} \cdot \bar{G}\bar{\Delta}\bar{\mathbf{u}} \\
\leq \frac{\varepsilon_f}{2} V + \frac{2\overline{m}^2\bar{G}^2\bar{\Delta}^2\bar{\mathbf{u}}^2}{\varepsilon_f\underline{m}} = \frac{\varepsilon_f}{2} V + \tilde{C}_d
\end{multline*}
by Young's inequality with $\varepsilon_f > 0$, where $\tilde{C}_d = 2\overline{m}^2\bar{G}^2\bar{\Delta}^2\bar{\mathbf{u}}^2/(\varepsilon_f\underline{m})$ and $\bar{\Delta} = \|\Xi - I\|$. The noise contribution (with $g_1 = 0$) yields:
\begin{equation*}
\frac{1}{2}\operatorname{Tr}\left[\sigma_2\sigma_2^{\top}\frac{\partial^2 V}{\partial \mathbf{b}^2}\right] \leq \frac{g_2^2}{2\underline{m}}\left(\varepsilon_c\overline{m}_{\mathbf{x}} + \frac{\overline{m}_{\mathbf{x}^2}}{2}\right)V + \underline{m}\tilde{C}_c
\end{equation*}
where $\tilde{C}_c = (\overline{m}/\underline{m} + \overline{m}_{\mathbf{x}}/(\varepsilon_c\underline{m}))g_2^2$. Combining all terms:
\begin{equation*}
\mathscr{L}V \leq -2\tilde{\gamma}_1 V + \underline{m}\tilde{C}_c + \tilde{C}_d
\end{equation*}
with $\tilde{\gamma}_1 = \tilde{\alpha} - \varepsilon_f/4 - (g_2^2/2\underline{m})(\varepsilon_c\overline{m}_{\mathbf{x}} + \overline{m}_{\mathbf{x}^2}/2) > 0$ for $\varepsilon_f > 0$. Applying Dynkin's formula and following the arguments in Theorem~\ref{thm:wasserstein_contraction} yields \eqref{eq:W2_FDI}.\qed
\end{pf}

In light of Theorem~\ref{thm:wasserstein_fdi}, for any horizon $\{t_0, \ldots, t_N\}$, one can investigate $\epsilon$-detectability between the nominal fault-free profile $\Theta_0$ and any faulty profile $\Theta_r$. The right-hand side of \eqref{eq:W2_FDI} in Theorem~\ref{thm:wasserstein_fdi} provides an upper bound on the Wasserstein separation between the nominal and faulty laws, which can be compared with empirical distances in simulations. empirical values of $\mathbb{W}_2$ strictly greater than zero then certify detectability for that specific fault scenario. In a similar manner, $\epsilon$-identifiability between any two fault profiles $\Theta_r$ and $\Theta_j$ can be established by performing the same analysis between their respective perturbed dynamics in~\eqref{eq:full_system_stochastic}. 

\section{Flow Matching on Trajectories for FDI} 
Now that we have established formal guarantees on fault detectability and identifiability via Theorem~\ref{thm:wasserstein_fdi}, we present a data-driven FDI method based on density flow matching along trajectories. Consider a set of $N_f+1$ fault profiles $\{\Theta_0, \ldots, \Theta_{N_f}\} \in \mathcal{F}$, where $\Theta_0$ is the nominal fault-free profile. Let $\xb^{\Theta_r}(t):=\{x^{\Theta_r}_{t_0}, \ldots, x^{\Theta_r}_{t_N}\} \in \R^n$ denote a trajectory of states from fault profile $\Theta_r$ for arbitrary timepoints $t_{t_0:t_N}:= \{t_0, \ldots, t_N\}$. We aim to deduce the fault profile by inferring the dynamics generating $\xb^{\Theta_r}(t)$ that best explains the observed trajectory. 

\subsection{Neural Flows}
In this subsection, we utilize flow matching to learn a vector field that transforms a simple base density into a target density. Let $\dot{\xb}^{\Theta_r} = \mathbf{F}_{\Theta_r}(\xb^{\Theta_r},t)$ from~\eqref{eq:full_system_deterministic} be the true dynamics generating the trajectory $\xb^{\Theta_r}(t)$ under fault profile $\Theta_r$. For a single timestep, let $p_0 = p(\xb, t=0)$ be a simple base density and $p_1 = p(\xb, t=1)$ be the target density from the FPE in~\eqref{eq:fokker_planck_equation}. The flow matching objective is to learn a vector field $\vb_{\boldsymbol{\theta}}(\xb,t)$ parameterized by neural network weights $\boldsymbol{\theta}$ such that the density $p(\xb,t)$ evolving under the learned dynamics matches the target density $p_1$ at time $t=1$ when initialized from $p_0$ at time $t=0$. The flow matching loss is given by
\begin{equation}
\begin{split}
\mathcal{L}_{\textrm{FM}}(\boldsymbol{\theta}) = \mathbb{E}_{t \sim \mathcal{U}(0,1), \ \xb^{\Theta_r} \sim p_t}[\|&\vb_{\boldsymbol{\theta}}(\xb, t) \\
&- \mathbf{F}_{\Theta_r}(\xb^{\Theta_r},t) \|^2 ].
\end{split}
\end{equation}
However, this loss is intractable since there is no closed-form expression for the true dynamics $\mathbf{F}_{\Theta_r}(\xb,t)$. To overcome this,~\cite{lipman2022flow, albergo2022building} propose using a conditional vector field $f(\xb^{\Theta_r}, t\mid \mathbf{c})$ where $\mathbf{c}$ is a latent variable sampled from some prior distribution $q(\mathbf{c})$, such that $p_t(\xb) = \mathbb{E}_{\mathbf{c} \sim q(\mathbf{c})}[p_t(\xb^{\Theta_r} \mid \mathbf{c})]$. In our setting, for Type 1 faults that model actuator faults with time-delayed offsets, $\mathbf{c} = [\boldsymbol{\eta}^\top, \mathbf{t}_{\textrm{start}}^{\top}]^\top$, and for Type 2 faults that model to simultaneous persistent unknown actuator and sensor faults, $\mathbf{c} = [\boldsymbol{\eta}^\top, \boldsymbol{\gamma}^{\top}]^\top$. This is known as Conditional Flow Matching (CFM) with the tractable loss
\begin{equation}
\label{eq:cfm_loss}
\mathcal{L}_{\textrm{CFM}}(\boldsymbol{\theta}) = \mathbb{E}_{t, \xb^{\Theta_r},  \mathbf{c} }\left[\left\|\vb_{\boldsymbol{\theta}}(\xb, t) - f(\xb^{\Theta_r}, t\mid \mathbf{c})\right\|^2\right],
\end{equation}
where the expectation is taken over $t \sim \mathcal{U}(0,1)$, $\xb^{\Theta_r} \sim p_t(\xb^{\Theta_r} | \mathbf{c})$, and $\mathbf{c} \sim q(\mathbf{c} )$. By minimizing this loss, as studied by ~\cite{lipman2022flow}, it can be shown that the learned vector field $\vb_{\boldsymbol{\theta}}(\xb,t)$ satisfies the FPE of~\eqref{eq:fokker_planck_equation} and thus the marginal density $p(\xb,t)$ evolving under $\vb_{\boldsymbol{\theta}}(\xb,t)$ matches the target density $p_1$ at time $t=1$. 

For each consecutive state pair $(\xb^{\Theta_r}_{t_k}, \xb^{\Theta_r}_{t_{k+1}})$ in the trajectory $\xb^{\Theta_r}(t)$, we adopt a Gaussian Bridge model (see \cite{zhang2024trajectory,albergo2022building}), where for any point $\xb_\tau$ between $\xb_{k}$ and $\xb_{k+1}$ we model the conditional density as $p(\xb_\tau \mid \xb_k, \xb_{k+1})\sim \mathcal{N}(\mu_{t_k}, \Sigma_{t_k})$ with
\begin{equation}
\begin{split}
\mu_{t_k} &= (1 - \tau)\xb^{\Theta_r}_{t_k} + \tau \xb^{\Theta_r}_{t_{k+1}}, \\
\Sigma_{t_k} &= \sigma_\mathrm{bridge}^2 \tau(1-\tau)\mathbf{I}_n,
\end{split}
\end{equation}
where $\tau \in [0,1]$, and $\sigma_\mathrm{bridge} > 0$ is a hyperparameter that controls the variance of the bridge, chosen to balance between noise regularization and trajectory fidelity. This construction also ensures that $\xb_k$ and $\xb_{k+1}$ match the dynamics~\eqref{eq:full_system_deterministic} almost surely while providing a smooth probabilistic interpolation between the two states. We parameterize a neural network $\mathcal{F}_{\boldsymbol{\theta}}: \mathbb{R}^{n_{\text {feat }}} \rightarrow \mathbb{R}^n \times \mathbb{R}^n$ that predicts the next state distribution conditioned on current and historical information. The model outputs a predicted mean and log-standard deviation
\begin{equation}
\left[\mu_{\boldsymbol{\theta}}, \log \sigma_{\boldsymbol{\theta}}\right]=\mathcal{F}_{\boldsymbol{\theta}}\left(\boldsymbol{\Phi}_{k, \tau}\right)
\end{equation}
where $\mu_{\boldsymbol{\theta}} \in \mathbb{R}^n$ and $\sigma_{\boldsymbol{\theta}} \in \mathbb{R}^n$ represent the diagonal Gaussian prediction for $\mathbf{x}_{k+1}$. Moreover, the feature vector $\boldsymbol{\Phi}_{k, \tau}$ at time step $k$ and interpolation parameter $\tau$ is constructed as
\begin{equation}
\boldsymbol{\Phi}_{k, \tau}=\left[t_k / t_N, \tau, \mathbf{x}_\tau^{\top}, \mathbf{x}_{k-1}^{\top}, \ldots, \mathbf{x}_{k-M}^{\top}, \mathbf{c}^\top\right]^{\top} ,
\end{equation}
where $t_k / t_N$ is the normalized absolute time, $\tau$ is the local interpolation time, $\xb_\tau$ is sampled from the Gaussian bridge, and $M_\mathrm{memory} \geq 1$ is the memory horizon. 

To effectively incorporate the fault conditioning vector $\mathbf{c}$ into the neural network, we utilize Feature-wise Linear Modulation (FiLM) layers developed by \cite{perez2018film} that apply affine transformations to intermediate feature maps based on $\mathbf{c}$. For a hidden activation $\mathbf{h} \in \mathbb{R}^{d_h}$ FiLM computes
\begin{equation}
\operatorname{FiLM}(\mathbf{h} \mid \mathbf{c})=\gamma(\mathbf{c}) \odot \mathbf{h}+\boldsymbol{\beta}(\mathbf{c}),
\end{equation}
where $\boldsymbol{\gamma}: \mathbb{R}^{d_c} \rightarrow \mathbb{R}^{d_h}$ and $\boldsymbol{\beta}: \mathbb{R}^{d_c} \rightarrow \mathbb{R}^{d_h}$ are learned neural networks that map the fault parameters to scale and shift vectors, respectively. This modulation allows the network to adapt its internal representations dynamically according to the fault characteristics encoded in $\mathbf{c}$. 

Since $\gamma(\mathbf{c})$ and $\beta(\mathbf{c})$ vary smoothly with $\mathbf{c}$, the model can interpolate between fault profiles during inference. This means that for any intermediate fault parameter $\mathbf{c}^{\boldsymbol{*}}=\lambda \mathbf{c}^{(\mathcal{i})}+(1-\lambda) \mathbf{c}^{(\mathcal{j})}$ with $\lambda \in[0,1]$, the network can generalize to unseen fault combinations that were not explicitly present in the training set. Consequently, we can train on a selected subset of fault profiles while maintaining generalizability across the continuous fault parameter space $\mathcal{F}=[0,1]^m \times[0,1]^p$.

\subsection{Neural Network Training and Fault Estimation}
Since we approximate each density by a Gaussian Bridge between consecutive states in the trajectory, the CFM loss in~\eqref{eq:cfm_loss} reduces to minimizing the negative log-likelihood (NLL) loss over a dataset of trajectories $\mathcal{D}=\left\{\left(\mathbf{x}^{(\Theta_j)}, \mathbf{c}^{(j)}\right)\right\}_{j=1}^{N_{\text {train }}}$. Thus, the loss function is
\begin{equation}
\mathcal{L}_{\mathrm{NLL}}(\boldsymbol{\theta})=\mathbb{E}\left[-\log p_{\boldsymbol{\theta}}\left(\xb_{k+1} \mid \xb_k, \mathbf{c}\right)\right].
\end{equation}
Since $p_{\boldsymbol{\theta}}\left(\xb_{k+1} \mid \xb_k, \mathbf{c}\right)$ is Gaussian with mean $\mu_{\boldsymbol{\theta}}$ and diagonal covariance $\Sigma_{\boldsymbol{\theta}} = \mathrm{diag}(\sigma_{\boldsymbol{\theta}}^2)$, the negative log-likelihood  loss can be derived as
\begin{equation}
\mathcal{L}_{\mathrm{NLL}}(\boldsymbol{\theta})=\mathbb{E}\left[\frac{1}{2} \sum_{\ell=1}^{n}\left(\frac{\left(\xb_{k+1}^{(\ell)}-\mu_{\boldsymbol{\theta}}^{(\ell)}\right)^2}{\sigma_{\boldsymbol{\theta}}^{(\ell) 2}}+\log \sigma_{\boldsymbol{\theta}}^{(\ell) 2}\right)\right].
\end{equation}
We also include a mean-squared error (MSE) regularizer to stabilize learning. Hence, the total loss function can be expressed by
\begin{equation}
\mathcal{L}(\boldsymbol{\theta})= \mathcal{L}_{\mathrm{NLL}}(\boldsymbol{\theta})+\lambda_{\mathrm{MSE}} \left\|\xb_{k+1}-\mu_{\boldsymbol{\theta}}\right\|^2 ,
\end{equation}
where $\lambda_{\mathrm{MSE}} > 0$ is a hyperparameter. 

Once the neural network is deployed, given an observed trajectory $\xb^{\mathrm{obs}}(t) = \{\xb_0^{\mathrm{obs}}, \ldots, \xb_K^{\mathrm{obs}}\}$, we evaluate the trained model across all fault profiles $\{\Theta_0, \ldots, \Theta_{N_f}\}$ by computing the total negative log-likelihood of the observed trajectory under each fault profile's conditioning vector $\mathbf{c}^{(r)}$. The fault profile that minimizes this negative log-likelihood is selected as the inferred fault, i.e., we find 
\begin{equation}
\hat{\Theta} = \arg \min_{\Theta} \sum_{k=0}^{K-1} -\log p_{\boldsymbol{\theta}}\left(\xb_{k+1}^{\mathrm{obs}} \mid \xb_k^{\mathrm{obs}}, \mathbf{c}^{(r)}\right),
\end{equation}
where each $\Theta_r$ corresponds to the conditioning vector $\mathbf{c}^{(r)}$. 

Together, these components establish a unified probabilistic formulation in which each fault hypothesis induces a distinct, parameterized density flow in which observed trajectories can be evaluated against. Because the mapping from fault parameters to density dynamics is continuous, the framework generalizes across the entire space of faults without requiring exhaustive training, enabling accurate detection and estimation across a wide range of actuator and sensor fault scenarios.
\section{Numerical Case Study}
We demonstrate the effectiveness of our FDI approach on a spacecraft attitude control system with four reaction wheels in a tetrahedral configuration. The system state is $x=\left[\theta^{\top}, \omega^{\top}, \omega_w^{\top}\right]^{\top} \in \mathbb{R}^{10}$ containing Euler angles (i.e., roll, pitch, and yaw), body angular velocities, and wheel speeds. We use the parameters outlined in~\cite{lee2017geometric} as follows: the spacecraft has inertia $I= \operatorname{diag}(1.0,1.0,0.8) \mathrm{kg} \cdot \mathrm{m}^2$, wheel inertia $J_w=0.01 \mathrm{~kg} \cdot \mathrm{~m}^2$, and uses a PD controller with gains $K_p=\operatorname{diag}(22.5,18.0,15.0)$ and $K_d=\operatorname{diag}(12.0,9.0,7.5)$ to track sinusoidal references over $T=60$ s horizons and a sampling time of $\Delta t=0.02 \mathrm{~s}$. The nominal tracking controller is designed to follow the reference trajectory $\theta_d(t)= [0.05 \sin (0.2 \pi t), 0.05 \cos (0.2 \pi t),(\pi / 250) t]^{\top}$ via the control law $u_{\text {nom }}=-K_p(\theta- \left.\theta_d\right)-K_d \omega$, which is mapped to wheel torques through $u_w=\operatorname{sat}\left(A^{\dagger} u_{\text {nom }}, 0.14\right)$, where $\operatorname{sat}(\cdot)$ denotes the saturation function, i.e., the maximum generated torque is $0.14\mathrm{~N.m}$. 

We consider two fault scenarios: Type 1 investigates actuator faults with time-delayed offsets through effectiveness parameters $\eta_i \in[0,1]$ for each wheel. Type 2 extends to persistent simultaneous actuator and sensor faults by adding sensor multiplicative fault parameters $\gamma_j \in[0,1]$ for seven measurements: three Euler angle sensors $\gamma_{1: 3}$  and four wheel speed sensors $\gamma_{4: 7}$, where the controller receives faulty measurements $\tilde{y}_\theta=\Gamma_\theta \theta$ with $\Gamma_\theta=\operatorname{diag}\left(\gamma_1, \gamma_2, \gamma_3\right)$ while gyroscopes are assumed to be fault-free. For Type 1 faults, each wheel transitions from nominal to degraded at randomly sampled times $t_{\text {start}, i} \sim \operatorname{Uniform}(8,42)$. Fault parameters are sampled from $\eta_i \sim \operatorname{Beta}(0.7,0.7)$ and $\gamma_j \sim \operatorname{Beta}(1.0,1.0)$ with $30-35 \%$ probability of nominal operation per channel. We generate $N_{\text {train }}=1000$ training and $N_{\text {val }}=200$ validation trajectories with additive Gaussian noise $(\sigma=0.001-0.002)$ plus initial condition uncertainty $\left(\sigma_0=0.01\right)$. 

Our conditional flow matching network $\mathcal{F}_{\boldsymbol{\theta}}$ takes input features $[\,t_k/t_N,\, \tau,\, \xb_\tau,\, \xb_{k-1:k-4},\, \mathbf{c}\,]$ where $\mathbf{c}$ adapts to the fault scenario.  For Type 1 time-delayed faults, $\mathbf{c}=\left[\boldsymbol{\eta}^{\top},\left(\mathbf{t}_{\text {start}} / t_N\right)^{\top}\right]^{\top} \in \mathbb{R}^8$ includes normalized onset times. For Type 2 persistent simultaneous actuator and sensor faults, $\mathbf{c}=\left[\boldsymbol{\eta}^{\top}, \boldsymbol{\gamma}^{\top}\right]^{\top} \in \mathbb{R}^{11}$ contains all actuator-sensor parameters. The architecture consists of two hidden layers of dimension 256 with FiLM conditioning, where each FiLM module applies learned affine transformations $h \mapsto h \cdot(1+\gamma(c))+\beta(c)$ based on fault parameters. The network outputs diagonal Gaussian parameters $\left(\mu_{\boldsymbol{\theta}}, \sigma_{\boldsymbol{\theta}}\right) \in \mathbb{R}^{10} \times \mathbb{R}^{10}$ for the next state prediction. We use a Gaussian bridge with variance $\sigma_{\text {bridge }}^2= 0.03^2$ and memory horizon $M_{\text {memory }}=4$ steps. Training uses Adam optimizer with learning rate $10^{-3}$ and batch size 256 over 15 epochs. 

For the deployed neural network $\mathcal{F}_{\boldsymbol{\theta}}$, given an observed trajectory $\xb^{\text{obs}}$ with unknown fault parameters, we estimate effectiveness by minimizing trajectory-level negative log likelihood. For Type 1 (i.e., the actuator faults scenario), we optimize 
$$\hat{\boldsymbol{\eta}}=\arg \min _{\boldsymbol{\eta}} \sum_{k=0}^{K-1}-\log p_{\boldsymbol{\theta}}\left(\xb_{k+1}^{\text {obs}} \mid \xb_k^{\text {obs}}, \mathbf{c}\right)+0.01 \sum_i(1- \left.\eta_i\right)^2$$
using Adam with 300 iterations starting from $\boldsymbol{\eta}^{(0)}=[0.95]^4$. For Type 2 simultaneous actuator and sensor faults, we jointly optimize 
\begin{equation*}
    \begin{split}
        (\hat{\boldsymbol{\eta}}, \hat{\boldsymbol{\gamma}})=& \arg \min _{\boldsymbol{\eta}, \boldsymbol{\gamma}} \sum_{k=0}^{K-1}-\log p_{\boldsymbol{\theta}}\left(\xb_{k+1}^{\text {obs}} \mid \xb_k^{\text {obs}},\mathbf{c}\right)\\
        &+ 0.01\left[\sum_i\left(1-\eta_i\right)^2+\sum_j\left(1-\gamma_j\right)^2\right],
    \end{split}
\end{equation*}
using Adam with 350 iterations from near-nominal initialization $\boldsymbol{\eta}^{(0)}=[0.9]^4$ and $\boldsymbol{\gamma}^{(0)}=[0.9]^7$. 

\subsection{Performance Evaluation}
Figure~\ref{fig:Type_1_FM} demonstrates the performance of our flow matching approach on actuator faults, i.e., Type 1. For actuator effectiveness estimation, the method achieves high accuracy with element-wise absolute errors of $|\boldsymbol{\eta}-\hat{\boldsymbol{\eta}}|= \left[0.012,0.039,0.054,0.059\right]$, corresponding to a mean parameter estimation error of approximately $4.1 \%$. The largest error occurs for $\eta_4$ at $5.9 \%$, likely due to its near-complete degradation $\left(\eta_4=0.15\right)$, which makes the parameter estimation more challenging. In Fig.~\ref{fig:Type_1_FM}, the actuator effectiveness estimates are obtained using the augmented Extended Kalman Filter (EKF) baseline, where fault parameters are added as additional parameters to the EKF (e.g., \cite{vettori2023adaptive}). Compared to our proposed FDI framework, the EKF exhibits slower convergence and reduced accuracy in estimating the fault magnitudes.

For the Type 2 scenario, we evaluate 10 different random actuator and sensor fault profiles $\Theta_{E_\mathcal{i}}$ for $\mathcal{i} \in \{0,\ldots, 9\}$, and Table~\ref{tb:flow_vs_ekf} presents comprehensive performance metrics on fault profile $\Theta_{E_1}=(\boldsymbol{\eta}=[0.3,0.9,0.7,1.0], \boldsymbol{\gamma}=[0.9,0.8,0.5,0.6,1.0,0.7,0.9])$. As shown in Table~\ref{tb:flow_vs_ekf}, our flow matching methodology achieves accurate parameter estimation in RMSE and $L_2$ Error where we compute the $\mathrm{RMSE} = \sqrt{\frac{1}{t_N} \int_0^{t_N}\|\Theta_{\textrm{true}}- \Theta_{\textrm{pred}}\|^2 dt}$ and the $L_2 \textrm{ Error} = \frac{1}{t_N} \int_0^{t_N}\|\Theta_{\textrm{true}}- \Theta_{\textrm{pred}}\| dt$ for all 10 fault hypotheses. For example, fault profiles $E_1,E_3, E_5,E_6$, and $E_9$ achieve an RMSE $= 0.000$ and $L_2$ Error $=0.000$. The confusion matrix results, shown in Fig.~\ref{fig:Type_2_FM} and Table~\ref{tb:flow_vs_ekf}, across 10 fault hypotheses with 10 noisy trajectories, demonstrate an overall classification accuracy of $67\%$, with a precision of $62.81\%$, and a false alarm rate of $3.67\%$. In comparison, the augmented EKF achieves an overall accuracy of $70\%$, a precision of $47\%$, and a high false alarm rate of $31\%$. Critically, the empirical Wasserstein distance $\mathbb{W}_2\left(\mu^{\Theta_{E_1}}, \mu^{\Theta_0}\right)=0.0717$ between the faulty and the nominal distributions is well below the theoretical upper bound of 0.7383 from Theorem~\ref{thm:wasserstein_fdi} (computed with $\bar{\Delta}_{\max }= \max _i\left|1-\eta_i\right|=0.7$), which validates our contraction-based detectability analysis. The KL divergence $\mathscr{D}_{K L}\left(\mu^{\Theta_{E_1}} \| \mu^{\Theta_0}\right)=44.119$ further confirms substantial distributional separation. 

\begin{figure}
\begin{center}
\includegraphics[width=\columnwidth]{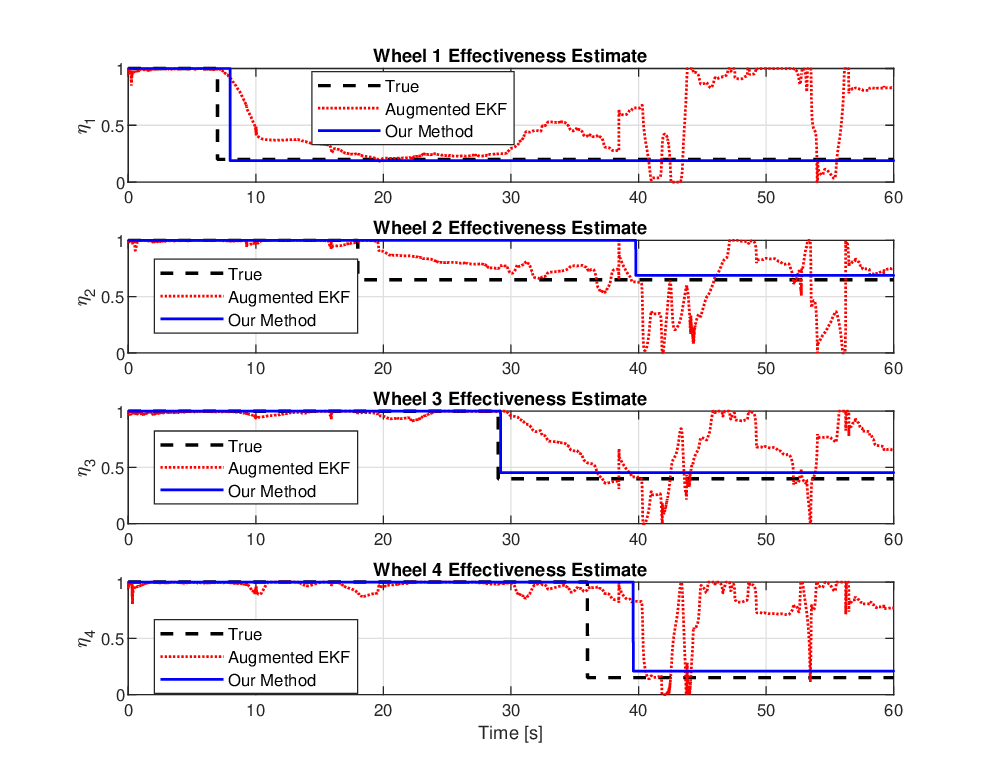}    
\caption{Plot of Type 1 faults comparing true  actuator effectiveness $\boldsymbol{\eta} = [0.20, 0.65, 0.40, 0.15]$ and onset times $\mathbf{t}_{\mathrm{start}} = [7.0, 18.0, 29.0, 36.0]$ with estimated $\hat{\boldsymbol{\eta}} = [0.188, 0.689, 0.454, 0.209]$ and $\hat{\mathbf{t}} = [8.1, 39.8, 29, 39.6]$ using our method and augmented EKF estimates.} 
\label{fig:Type_1_FM}
\end{center}
\end{figure}


\begin{figure}
\begin{center}
\includegraphics[width=9.5cm]{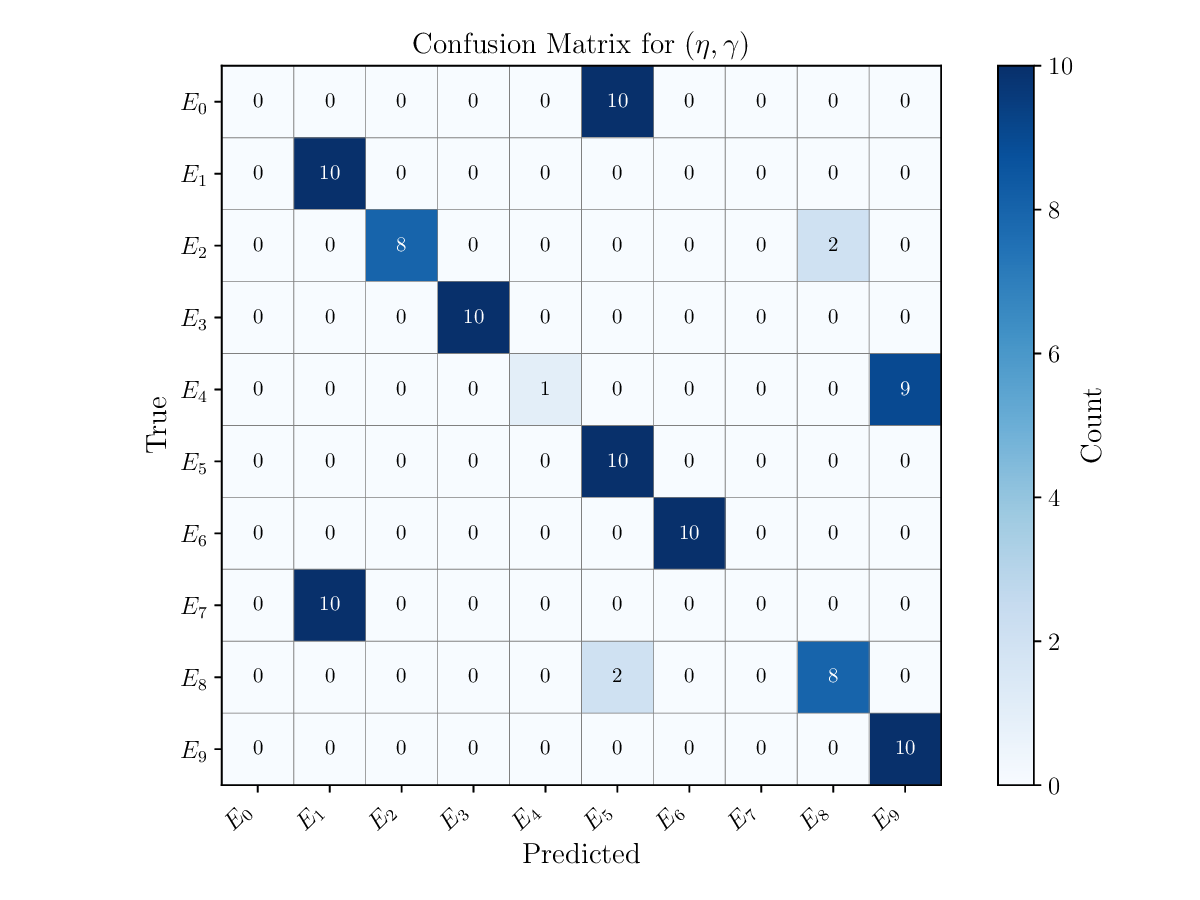}    
\caption{Plot of confusion matrix of Type 2 faults comparing 10 different random fault hypotheses with 10 trajectories for each with noise.} 
\label{fig:Type_2_FM}
\end{center}
\end{figure}

\begin{table}[hb]
\begin{center}
\caption{Performance Metrics for the Flow Matching-Based FDI Method}\label{tb:flow_vs_ekf}
\begin{tabular}{lcc}
\textbf{Metric} & \textbf{Value} \\ \hline

$\mathscr{D}_{KL}\left(\mu^{\Theta_{E_1}}\| \ \mu^{\Theta_0}\right)$ & $44.1190$  \\
$\mathbb{W}_2\left(\mu^{\Theta_{E_1}}, \mu^{\Theta_0}\right)$ & $0.0717$  \\
Theoretical $\mathbb{W}_2$ Bound ($\bar\Delta_{\max}$) & $0.7383$  \\

Overall Accuracy (All $E_\mathcal{i}$) & $67.00\%$  \\
Avg Precision (All $E_\mathcal{i}$) & $62.81\%$ \\
Avg False Alarm Rate (All $E_\mathcal{i}$) & $3.67\%$ \\
Overall RMSE (All $E_\mathcal{i}$) & $0.3104$  \\ 
Overall L2 (All $E_\mathcal{i}$) & $0.2629$  \\ \hline

\end{tabular}
\end{center}
\end{table} 

\section{Conclusion}
In this paper, we studied and developed a data-driven fault detection and identification (FDI) method for nonlinear control-affine systems. The proposed method employs a probabilistic approach, where the impact of faults is modeled through a trajectory-level density evolution. Moreover, explicit bounds on the detectability and identifiability of faults were provided through the Wasserstein metric. Our approach combines theoretical bounds derived via contraction analysis with a practical data-driven implementation using conditional flow matching to learn fault-conditioned probability flows, thereby facilitating generalization to unseen fault combinations. Evaluation of our FDI approach on a nonlinear spacecraft attitude control system showed reliable actuator and sensor fault estimation and a lower false alarm rate compared to an augmented Extended Kalman Filter (EKF) FDI method. We also showed that the empirical Wasserstein separations are consistent with theoretical predictions. 

\bibliography{ifacconf}             
                                                   







\end{document}